\documentclass[a4paper,11pt]{article}
\usepackage[utf8]{inputenc}
\usepackage[T1]{fontenc}
\usepackage{amsfonts,epsf,amsmath,amssymb,amsthm}
\usepackage{tgtermes}
\usepackage{enumerate}
\usepackage{tikz}
\usepackage[hyphens]{url}
\usepackage[hidelinks]{hyperref}
\usepackage{fge}
\usepackage[margin=1in]{geometry}

\renewcommand{\theenumi}%
{\arabic{enumi}}
\renewcommand{\theenumii}%
{\theenumi.\arabic{enumii}}
\renewcommand{\theenumiii}%
{\theenumii.\arabic{enumiii}}
\renewcommand{\labelenumii}%
{\theenumii.}

\newtheorem{theorem}{Theorem}[section]
\newtheorem{lemma}[theorem]{Lemma}
\newtheorem{proposition}[theorem]{Proposition}
\newtheorem{corollary}[theorem]{Corollary}

\renewenvironment{proof}[1][Proof.]{\begin{trivlist}
\item[\hskip \labelsep {\bfseries #1}]}{\end{trivlist}}

\def\qed{\ifhmode\unskip\nobreak\fi\hfill
  \ifmmode\square\else$\square$\fi}

\newcommand{\RR}{\mathbb{R}}
\newcommand{\RPZ}{\RR^+_0}
\newcommand{\NN}{\mathbb{N}}
\newcommand{\ZZ}{\mathbb{Z}}

\newcommand{\cclass}[1]{\mbox{{\sf #1}}}
\newcommand{\PP}{\cclass{P}}
\newcommand{\EXP}{\cclass{EXP}}

\newcommand{\cproblem}[1]{\textrm{#1}}
\newcommand{\halt}[1]{\cproblem{HALT}_{#1}}

\newcommand{\halto}{\cproblem{HALT}_{\epsilon}^1}

\newcommand{\oo}{\operatorname{o}}
\newcommand{\OO}{\operatorname{O}}
\newcommand{\OOm}{\operatorname{\Omega}}

\newcommand{\FF}{\mathcal{F}}

\newcommand{\ie}{i.e. }
\newcommand{\eg}{e.g. }
\newcommand{\tup}[3]{#1, #2 \ldots #3}
\newcommand{\seq}[2]{#1, #2\ldots}
\newcommand{\step}[1]{{#1}}

\newcommand{\blank}{\fgeupbracket}
\newcommand{\state}[1]{\textsc{#1}}
\newcommand{\qz}{q_{\state{0}}}
\newcommand{\qa}{q_{\state{acc}}}
\newcommand{\qr}{q_{\state{rej}}}

\title{Verifying Time Complexity of Deterministic Turing Machines\footnote{This work is partially funded by the Slovenian Reaserarch Agency.}}
\date{\today}
\author{David Gajser\\ IMFM, Jadranska 19, 1000 Ljubljana, Slovenija \\\href{mailto:david.gajser@fmf.uni-lj.si}{david.gajser@fmf.uni-lj.si}}

\begin{document}
\maketitle
\thispagestyle{empty}
\paragraph{Abstract.} 
We show that, for all reasonable functions $T(n)=\oo(n\log n)$, we can algorithmically verify whether a given one-tape Turing machine runs in time at most $T(n)$. This is a tight bound on the order of growth for the function $T$ because we prove that, for $T(n)\geq(n+1)$ and $T(n)=\Omega(n\log n)$, there exists no algorithm that would verify whether a given one-tape Turing machine runs in time at most $T(n)$.


We give results also for the case of multi-tape Turing machines. We show that we can verify whether a given multi-tape Turing machine runs in time at most $T(n)$ iff $T(n_0)< n_0+1$ for some $n_0\in\NN$.

We prove a very general undecidability result stating that, for any class of functions $\FF$ that contains arbitrary large constants, we cannot verify whether a given Turing machine runs in time $T(n)$ for some $T\in\FF$. In particular, we cannot verify whether a Turing machine runs in constant, polynomial or exponential time.


%

\newpage
\pagenumbering{arabic}
\section{Introduction}

It is tempting to argue about a Turing machine's time complexity. However, from the undecidability of the Halting problem we know that we can not algorithmically tell even whether a given Turing machine halts on a given input. Can we perhaps check whether it is of a specified time complexity? While the answer is NO in most cases, there is an interesting case where the answer is YES: verifying a time bound $T(n)=Cn+D$, $C,D\in\ZZ$, for a given one-tape Turing machine.

The inspiration for most undecidability results in this paper was the answer of Emanuele Viola on the forum Theoretical Computer Science Stack Exchange~\cite{stack}. From his answer it is easy to see that, for any $k\in\NN$, there is no algorithm that would decide whether a given Turing machine runs in time $\OO(n^k)$. This holds even for $k=0$, thus we cannot algorithmically verify whether a Turing machine runs in constant time (see Theorems~\ref{multiO} and~\ref{oneO} for a general result on this topic). But what if the time bound is explicitly given with a function $T(n)$, \ie can we verify whether a given Turing machine runs in time at most  $T(n)$?

The answer is not always negative in this case and it also depends on the model of the given Turing machine. To further discuss this question, let us denote the problems of whether a  [multi-tape, one-tape] Turing machine\footnote{All Turing machines in this article are deterministic.} runs in time at most $T(n)$ with [$\halt{T(n)}$, $\halt{T(n)}^1$].

First, we observe that no matter what model we use, we can always check whether a Turing machine runs in times like $T(n)=n^3$ and $T(n)=5$. In the first case, $T(0)=0$ and thus there exist no Turing machine that would run in time $T(n)$; otherwise the machine would make $0$ steps on the empty input, but all machines make at least one step on each input. In the case $T(n)=5$ we only need to verify the running time of a Turing machine for inputs of length at most 5 to see whether it runs in time $T(n)$, which can easily be done (see  Lemma~\ref{pozMulti} for details). Actually, we will see in Theorem~\ref{iff} and Proposition~\ref{poz1} that if $T(n_0)<n_0+1$ for some $n_0\in\NN$, then the problems $\halt{T(n)}$ and $\halt{T(n)}^1$ are decidable.

The case where $T(n)\geq n+1$, for all $n\in\NN$, is more interesting, especially in the one-tape case. While for multi-tape machines the inequality $T(n)\geq n+1$ already suffices for undecidability of $\halt{T(n)}$ (see Theorem~\ref{iff}), we have to add the condition $T(n)=\OOm(n\log n)$ to prove undecidability of $\halt{T(n)}^1$ (see Theorem~\ref{neg1}). But is this really necessary?

First notice that we can decide $\halt{n+1}^1$, because the head of a one-tape Turing machine $M$ that runs in time $(n+1)$ can move to the left only in the last step (when $M$ goes to a halting state) or in the first step (in this case $M$ has to halt in the next step since $M$ makes at most 2 steps on inputs of length 1). It seems plausible that by similar inference we could show that $\halt{n+D}^1$ is decidable for each $D\in\NN$. But what if we add a multiplicative constant, \ie is $\halt{Cn+D}^1$ decidable for $C,D\in\NN$?

We show not only that this problem is decidable, but also that all problems $\halt{T(n)}^1$ are decidable for tangible enough functions $T(n)=\oo(n\log n)$ (see Theorem~\ref{main}). Thus we give a sharp bound on how fast the function $T(n)$ can increase so that $\halt{T(n)}^1$ is decidable.

We also show that every one-tape Turing machine $M$ that runs in time $T(n)=\oo(n\log n)$ runs in linear time (see Corollary~\ref{time}). It is known~\cite{Hartmanis, Trakhtenbrot} that such $M$ accepts a regular language and as an upgrade of this result we describe an algorithm that constructs an equivalent deterministic finite automaton out of $M$, if integers $C,D$ such that $M$ runs in $Cn+D$ are also given as inputs (see Theorem~\ref{constructive}).

A main tool used in this paper for the analysis of one-tape Turing machines are crossing sequences. They were first studied in1960s by Hartmanis~\cite{Hartmanis}, Hennie~\cite{hennie} and Trakhtenbrot~\cite{Trakhtenbrot}.
In 1968 Hartmanis~\cite{Hartmanis} proved that any one-tape Turing machine which runs in time $\oo(n \log n)$ recognizes a regular language\footnote{Hartmanis acknowledges that Trakhtenbrot~\cite{Trakhtenbrot} came to the same result independently.}. To be more precise, he showed that a one-tape Turing machine which runs in time $\oo(n \log n)$ produces only crossing sequences of bounded length and then he used Hennie's~\cite{hennie} result which tells that such Turing machines recognize only regular languages. 
Later in 1980s Kobayashi~\cite{kobayashi} gave another proof of the same result but, in contrast to Hartmanis' approach, his proof gives a way to compute an upper bound on the length of crossing sequences. This constructiveness is essential in the proof of our Theorem~\ref{main} which states that the problem $\halt{T(n)}^1$ is decidable for nice functions $T(n)=\oo(n \log n)$.

Kobayashi's result was actually a bit stronger than the one by Hartmanis, since Kobayashi showed regularity of every language that is \emph{accepted} in $\oo(n \log n)$ time by some (deterministic) one-tape Turing machine, \ie his $\oo(n \log n)$ bound concerns only the accepting computations. 
In 2009 Pighizzini~\cite{pighizzini} further extended Kobayashi's result by showing regularity of some languages accepted by nondeterministic Turing machines.

The recent paper by Tadaki, Yamakami and Lin~\cite{summary} summarizes results about one-tape linear-time Turing machines of different types, from deterministic to quantum.



\section{Preliminaries}
	\label{prelim}

\paragraph{Basic notation.} 
Let $\NN$ be the set of non-negative integers, let $\RR^+$ be the set of positive real numbers and let $\RPZ$ be the set of non-negative real numbers. For $r\in\RPZ$, we use $\lfloor r\rfloor$ to denote the integer part of $r$ and $\lceil r\rceil$ to denote the smallest integer greater than or equal to $r$. For a function $f:\NN\rightarrow\RPZ$, we define $\lfloor f\rfloor:\NN \rightarrow\NN$, $\lfloor f\rfloor (n)=\lfloor f(n)\rfloor$. 

All logarithms with no base written have base 2. We use $\epsilon$ for the empty word and $|w|$ for the length of word $w$. For words $w_1$ and $w_2$ let $w_1w_2$ denote their concatenation. 

For functions $f,g:\NN\rightarrow\RPZ$, we say that $\left[f(n)=\OO(g(n)),\ f(n)=\OOm(g(n))\right ]$ if there exist $k>0$ and $n_0\in\NN$ such that, for all $n\geq n_0$, it holds $\left[f(n)\leq k\cdot g(n),\ f(n)\geq k\cdot g(n)\right]$. We say that $f(n)=\oo(n\log n)$ if ${\displaystyle \lim_{n\to\infty}}\frac{f(n)}{n\log n}=0$.

Let $f:\NN\rightarrow\NN$ be a function. If there exists a multi-tape Turing-machine $M$ that halts on all inputs and at the end of computation on any input $w$ has $f(w)$ written on the second tape, then we say that $M$ \emph{computes} $f$ and that $f$ is \emph{computable}. Since we will be using the Church-Turing thesis, we will say that a function $f$ is computable iff there exists an algorithm which tells us how to compute it.

For a function $f:\NN\rightarrow\RPZ$, we say that $f$ \emph{computably converges to $\infty$} if, for each $K\in\NN$, we can construct $n_K\in\NN$ (i.e. the function $K\mapsto n_K$ is computable) such that, for all $n\geq n_K$, it holds $f(n)\geq K$.

\paragraph{Models used.} We use two different kinds of deterministic Turing machines as a model of computation: a multi-tape Turing machine and a one-tape Turing machine. A \emph{multi-tape Turing machine} is an 8-tuple $M=(Q,\Sigma,\Gamma,\blank, \delta,\qz,\qa,\qr)$, where $Q$ is a set of states, $\Sigma\neq\emptyset$ an input alphabet, $\Gamma\supseteq\Sigma$ a tape alphabet, $\blank\in\Gamma\backslash \Sigma$ a blank symbol, $\delta$ a transition function (defining also the number of tapes) and $\qz,\qa,\qr\in Q$ pairwise distinct starting, accepting and rejecting states. The machine has a read-only input tape and all of its tapes are both-way infinite\footnote{The results in this paper do not depend on these two parameters.}. There is no output tape included, so on input $w$ there are three possible outcomes of the computation: 
\begin{itemize}
\item $M$ stops in state $\qa$. We say that $M$ accepts the input $w$ and write $M(w)=1$,
\item $M$ stops in state $\qr$. We say that $M$ rejects the input $w$ and write $M(w)=0$,
\item $M$ does not halt. We say that $M$ runs forever and write $M(w)=\infty$.
\end{itemize}

A \emph{one-tape Turing machine} is again an 8-tuple $M=(Q,\Sigma,\Gamma,\blank, \delta,\qz,\qa,\qr)$, where the transition function $\delta$ deals only with one both-way infinite read-write tape, which is also the input tape. For the simplicity we assume that at each step (\ie iteration of transition function) the head moves (left or right). 
Just like above, there are three possible results of the computation.

For a function $T:\NN\rightarrow \RPZ$, if a $\left[\textrm{one-tape, multi-tape}\right ]$ Turing machine $M$, for each $n\in\NN$, makes \textbf{at most} $T(n)$ steps on inputs of length $n$, then we say that  $M$ \emph{runs in time $T(n)$}.


\paragraph{Finite automata.}
In the next proposition we state a known fact about finite automata that will be needed in the proofs of Theorem~\ref{main} and Theorem~\ref{constructive}. All the material for its proof can be found \eg in~\cite[Chapter 1]{sipser}.
\begin{proposition}
	\label{regbasic}
Let $A$ and $B$ be deterministic finite automata that recognize languages $L_A$ and $L_B$. Then there exists an algorithm, which given $A$ and $B$, constructs deterministic finite automata $C_1$, $C_2$ and $C_3$, such that 
\begin{itemize}
\item $C_1$ recognizes the language $L_A\bigcup L_B$,
\item $C_2$ recognizes the language $L_A L_B$ and
\item $C_3$ recognizes the language $L_A^*$.
\end{itemize}
\end{proposition}
\paragraph{Manageable functions.}
We say that a function $f:\NN\rightarrow\RPZ$ is \emph{manageable}\footnote{Manageable functions will appear in Theorem~\ref{main}} if there exists an algorithm, which given $\tup{A_0}{A_1}{A_k}\in\NN\backslash\{0\}$ and $\tup{B_0}{B_1}{B_k}\in\NN$, decides whether the inequality $$f(A_0+x_1A_1+x_2A_2+\cdots+x_kA_k)<B_0+x_1B_1+x_2B_2+\cdots+x_kB_k$$ holds for some $\tup{x_1}{x_2}{x_k}\in\NN$. 

Note that there are only integers on the right-hand side of the inequality. Thus the following holds.

\begin{proposition}
	\label{integer}
A function $f:\NN\rightarrow\RPZ$ is manageable iff its integer part $\lfloor f\rfloor$ is manageable.
\end{proposition}

The next proposition gives some examples of manageable functions.
\begin{proposition}
	\label{examples}
Let $f:\NN\rightarrow\NN$  be a computable function. If
\begin{itemize}
\item $f$ is linear (i.e. of the form $Cn+D$) or
\item $\frac{f(n)}{n}$ computably converges to $\infty$,
\end{itemize}
then $f$ is manageable.
\end{proposition}
\begin{proof}
The case when $f$ is linear is easy and is left for the reader, so suppose that $\frac{f(n)}{n}$ computably converges to $\infty$. The next algorithm proves manageability of $f$:
\begin{itemize}
\item  Let $\tup{A_0}{A_1}{A_k}\in\NN\backslash\{0\}$ and $\tup{B_0}{B_1}{B_k}\in\NN$ be given. 
\item Find $C\in\NN$ such that, for all $i=\tup{0}{1}{k}$, it holds $C A_i\geq B_i$.
\item Find $n_C$ such that  $f(n)\geq Cn$ for all $n\geq n_C$. This can be done because $\frac{f(n)}{n}$ computably converges to $\infty$.
\item For $i=\tup{1}{2}{k}$, let  $y_i\in\NN$ be such that $A_0+y_iA_i\geq n_C$.

It follows that the inequality  $$f(A_0+x_1A_1+x_2A_2+\cdots+x_kA_k)\geq B_0+x_1B_1+x_2B_2+\cdots+x_kB_k$$ holds for $\tup{x_1}{x_2}{x_k}\in\NN$ if there exists an index $i$ such that $x_i\geq y_i$.

Indeed, $x_i\geq y_i$ implies $A_0+x_1A_1+x_2A_2+\cdots+x_kA_k\geq n_C$, which implies $f(A_0+x_1A_1+x_2A_2+\cdots+x_kA_k)\geq C( A_0+x_1A_1+x_2A_2+\cdots+x_kA_k)$.
\item Check if the inequality $$f(A_0+x_1A_1+x_2A_2+\cdots+x_kA_k)<B_0+x_1B_1+x_2B_2+\cdots+x_kB_k$$ holds for some non-negative integers $\tup{x_1<y_1}{x_2<y_2}{x_k<y_k}$.\qed
\end{itemize}
\end{proof}

We just proved (using also Proposition~\ref{integer}) that $n,\ 3n+2,\ n\sqrt{\log n},\ n^2,\ 2^n$ are all manageable functions. The next proposition tells us that integer part of a manageable functions cannot be too complicated.
\begin{proposition}
	\label{compf}
An integer part $\lfloor f \rfloor$ of a manageable function $f:\NN\rightarrow\RPZ$ is a computable function.
\end{proposition}
\begin{proof}
For $n\in\NN$, the following algorithm computes $\lfloor f\rfloor (n)$:
\begin{itemize}
\item If $n=0$ return $\lfloor f(0)\rfloor$. Else, return the largest $i$ for which $f(n)\geq i$.\qed
\end{itemize}
\end{proof}

\paragraph{Decidability.}
When talking about decidability, we will define problems as sets of inputs. We will not give the actual encodings\footnote{We assume that we have some canonical encoding of $\left[\textrm{one-tape, multi-tape}\right ]$ Turing machines and pairs of words.}, thus we will have to tell what all possible inputs are. We will do this using ``$\subseteq$'', as can be seen in the definition of a variant of the famous halting problem:
\begin{align*}
\halt{}&=\{(M,w);\textrm{ a one-tape Turing machine }M\textrm{ halts on input }w\}\\
&\subseteq\{(M,w);\ w\textrm{ is an input of a one-tape Turing machine }M\}.
\end{align*}

 We will say that a multi-tape Turing machine $M$ decides (or solves) problem $L\subseteq \bar{L}$, if 
$M$ halts on all inputs that represent elements of $\bar{L}$ and accepts exactly those inputs that represent elements of $L$.
We will say that a problem is decidable, if there exists a multi-tape Turing machine that decides it and undecidable otherwise. When proving decidability of some problem $L\subseteq\bar{L}$, we will only describe an algorithm which solves $L\subseteq \bar{L}$, not giving the construction of a multi-tape Turing machine for this purpose. Thus we will be using the Church-Turing thesis and will confound decidability with the existence of algorithm. 

It is well known that the halting problem is undecidable. We can reduce it to (consequently undecidable) problem
\begin{align*}
\halto&=\{\textrm{one-tape Turing machines that halt on input }\epsilon\}\\
&\subseteq\{\textrm{one-tape Turing machines}\},
\end{align*}
which will prove very useful for later reductions.
\begin{lemma}
	\label{ulema}
The problem $\halto$ is undecidable.
\end{lemma}
\begin{proof}
Suppose that $\halto$ is decidable. Then there exists a multi-tape Turing machine $H_\epsilon$ that decides $\halto$. Let $H$ be an arbitrary one-tape Turing machine and $h$ its input. We will describe an algorithm to decide whether $H$ halts on input $h$, thus contradicting the undecidability of $\halt{}$:
\begin{itemize}
\item Construct a one-tape Turing machine $\tilde{H}$ that on empty input first writes $h$ on its tape and then simulates $H$ on $h$. 
\item Return $H_\epsilon(\tilde{H})$.
\end{itemize}
Thus $\halto$ is undecidable. \qed
\end{proof}


\paragraph{Definitions of problems.}
Next we define the problems that will be in our interest in this article. For a function $T:\NN\rightarrow\RPZ$, define the problems \begin{align*}
\halt{T(n)}=&\{\textrm{multi-tape Turing machines that run in time }T(n)\}\\
\subseteq&\{\textrm{multi-tape Turing machines}\},\\
\halt{T(n)}^1=&\{\textrm{one-tape Turing machines that run in time }T(n)\}\\
\subseteq&\{\textrm{one-tape Turing machines}\}.
\end{align*}
Note that there is no big $\OO$ notation in the above definition, which is handled next. For a class of functions $\FF\subseteq\{T:\NN\rightarrow\RPZ\}$, define 
\begin{align*}
\halt{\FF}=&\{\textrm{multi-tape Turing machines that run in time }T(n)\textrm{ for some }T\in \FF\}\\
\subseteq&\{\textrm{multi-tape Turing machines}\},\\
\halt{\FF}^1=&\{\textrm{one-tape Turing machines that run in time }T(n)\textrm{ for some }T\in \FF\}\\
\subseteq&\{\textrm{one-tape Turing machines}\}.
\end{align*}

If $\FF$ is the class of polynomials, we write $\FF=\PP$, thus $\halt{\PP}$ is the set of all multi-tape polynomial-time Turing machines. If $\FF$ is the class of exponential functions, we write $\FF=\EXP$, thus $\halt{\EXP}$ is the set of all multi-tape exponential-time Turing machines. For a function $T:\NN\rightarrow\RPZ$, if $\FF=\{f:\NN\rightarrow\RPZ;f(n)=\OO(T(n))\}$,  we write $\FF=\OO(T(n))$, thus $\halt{\OO(T(n))}$ is the set of all multi-tape Turing machines that run in $\OO(T(n))$ time.

Problems $\halt{\PP}^1$,  $\halt{\EXP}^1$ and $\halt{\OO(T(n))}^1$ are defined similarly for one-tape Turing machines.


In contrast to already defined problems, we are sometimes interested in problems $L\subseteq\bar{L}$ where $$\bar{L}=\{\textrm{[one-tape, multi-tape] Turing machines that always halt}\}.$$ This problems are motivated by  a belief that ``programmers know that their program will always terminate''. The negative results\footnote{\ie those that show undecidability} in this paper do not address this problem directly, but with a slight modification they could.
\section{Multi-tape machines}
	\label{multi}

We start with multi-tape Turing machines because the results here are simpler than in one-tape case and the reader may get a better feeling of what is going on.
This section includes the results about decidability of problems $\halt{\FF}$ and $\halt{T(n)}$ for different $\FF$ and $T$. We will see that all ``basic'' problems $\halt{\FF}$ are undecidable and we will prove a tight bound on the function $T$ for which the problem $\halt{T(n)}$ is decidable. 

Let us start with the positive result.
\begin{lemma}
	\label{pozMulti}
Let $T:\NN\rightarrow\RPZ$ be a function such that, for some $n_0\in\NN$, it holds $T(n_0)< n_0+1$. Then the problem $\halt{T(n)}$ is decidable.
\end{lemma}
\begin{proof}
Let $n_0$ be such that $T(n_0)< n_0+1$ and let $M$ be an arbitrary multi-tape Turing machine. We will describe an algorithm to decide whether $M$ runs in time $T(n)$, thus proving decidability of $\halt{T(n)}$:
\begin{itemize}
\item First, check if the running time of $M$ on inputs of lengths $n\leq n_0$ is at most $T(n)$. If not, \textbf{return 0}. Else, let $T_w$ be the maximum number of steps $M$ makes on inputs of length $n_0$ and suppose this maximum is achieved on input $w$.
%
\item If $T_w\leq T(n)$ for all $n>n_0$, \textbf{return 1}. Else, \textbf{return 0}.
\end{itemize}

To prove finiteness and correctness  of the algorithm, note that if $M$ makes at most $T(n_0)$ steps for all inputs of size $n_0$, then $M$ never reads the $(n_0+1)$-st bit of any input. 
In this case $M$ makes at most $T_w$ steps on inputs of size more than $n_0$. Moreover, for each $n\geq n_0$, there exists an input of length $n$ on which $M$ makes exactly $T_w$ steps (all inputs that begin with $w$ are such). There are only finitely many possibilities for $T_w$ because $T_w\leq T(n_0)$, thus the last line of the algorithm can be done in constant time.\qed
\end{proof}

The next lemma is the heart of negative results. Its one-tape analog, Lemma~\ref{logLema}, is much more technical.
\begin{lemma}
	\label{logLema2}
Let $T:\NN\rightarrow\RPZ$ be a function such that, for all $n\in\NN$, it holds $T(n)\geq n+1$. 
Then there exists an algorithm that takes as input a one-tape Turing machine $H$ and returns returns a multi-tape Turing machine $\tilde{H}$ such that
$$H(\epsilon)=\infty\textrm{ \textbf{iff} }\tilde{H}\textrm{ runs in time }(n+1)\textrm{ \textbf{iff} }\tilde{H}\textrm{ runs in time }T(n)\textrm{ \textbf{iff} }\tilde{H}\textrm{ always halts.}$$
\end{lemma}
\begin{proof}
Given $H$, we can construct a multi-tape Turing machine $\tilde{H}$ that uses the input tape just for counting steps. On input $w$, $\tilde{H}$ simulates $(|w|+1)$ steps of ($H$ on input $\epsilon$), while itself making exactly $(|w|+1)$ steps before halting. If $H$ halts in less than $(|w|+1)$ steps, $\tilde{H}$ starts an infinite loop.\qed
\end{proof}

\begin{theorem}
	\label{iff}
The problem $\halt{T(n)}$ is undecidable iff $T(n)\geq n+1$ holds for all $n\in\NN$.
\end{theorem}
\begin{proof}
Lemma~\ref{pozMulti} proves the only if part. For the if part, let a function $T:\NN\rightarrow\RPZ$ be such that  $T(n)\geq n+1$ for all $n\in\NN$

Suppose that $\halt{T(n)}$ is decidable. For an arbitrary one-tape Turing machine $H$, let $\tilde{H}$ be a multi-tape Turing machine that runs in time $T(n)$ iff $H(\epsilon)=\infty$. By Lemma~\ref{logLema2} we can construct  it from $H$. Since $\halt{T(n)}$ is decidable, we can verify whether $\tilde{H}$ runs in time $T(n)$ and thus we can verify whether $H$ halts on input $\epsilon$. This is a contradiction because $\halto$ is undecidable after Lemma~\ref{ulema}.\qed
\end{proof}

In the beginning of this section we declared that  all ``basic'' problems $\halt{\FF}$ are undecidable. The next theorem with its corollaries reveals what was meant with ``basic''.
\begin{theorem}
	\label{multiO}
Let $\FF\subseteq\{f:\NN\rightarrow\RPZ\}$ be a class of functions that contains arbitrary large constants. Then the problem $\halt{\FF}$ is undecidable.
\end{theorem}

\begin{proof}
The proof is by contradiction, so let a multi-tape Turing machine $H_{\FF}$ decide $\halt{\FF}$ for some class $\FF$ of functions that contains arbitrary large constants. Define the class  $F=\{T\in \FF,\ T(n)\geq n+1\textrm{ for all }n\}$ and consider two separate cases:
\begin{enumerate}
\item If $F$ is empty, then the following algorithm solves $\halto$:
\begin{itemize}
\item Given a one-tape Turing machine  $H$, construct a multi-tape Turing machine $\tilde{H}$ that on input $w$
\begin{itemize}
\item simulates $|w|$ steps of ($H$ on input $\epsilon$).
\item If $H$ halts in less than $|w|$ steps, then $\tilde{H}$ also halts and does not make any additional steps. 
\item Else, $\tilde{H}$ makes at least one additional arbitrary step and halts. 
\end{itemize}

We do not need $\tilde{H}$ to efficiently simulate $H$, but it is necessary that $\tilde{H}$ runs in constant time if $H$ halts on $\epsilon$.
\item Return $H_\FF(\tilde{H})$.
\end{itemize}

It is clear from the algorithm that $H$ halts on input $\epsilon$ \textbf{iff} $\tilde{H}$ runs in constant time \textbf{iff} $\tilde{H}$ runs in time $\tilde{T}(n)$ for some $\tilde{T}(n)\in\FF$.  The last \textbf{iff} follows from $F=\emptyset$ (as seen in the proof of Lemma~\ref{pozMulti}, a Turing machine that runs in time ${T}(n)$ such that  ${T}(n_0)<n_0+1$  for some $n_0\in\NN$, runs in constant time) and the fact that $\FF$ contains arbitrary large constants.
%
%
%
It follows that the above algorithm decides $\halto$, which is a contradiction.
\item If $F$ is not empty, then we can (again) solve $\halto$:

For an arbitrary one-tape Turing machine $H$, use Lemma~\ref{logLema2} to construct a multi-tape Turing machine $\tilde{H}$ which runs in time $(n+1)$ \textbf{iff} $\tilde{H}$ always halts \textbf{iff} $H(\epsilon)=\infty$.
Now $\left(1-H_{\FF}(\tilde{H})\right)$ tells whether  $H$ halts on input $\epsilon$.

Since $\halto$ is undecidable, we have a contradiction.
\end{enumerate}

To sum up, $\halt{\FF}$ is undecidable.\qed
\end{proof}
This theorem has some interesting corollaries.
\begin{corollary}
The problems  $\halt{\PP}$, $\halt{\EXP}$ and $\halt{\OO(T(n))}$ for $T(n)=\OOm(1)$ are undecidable.
\end{corollary}
\begin{proof}
The classes $\PP$, $\EXP$ and $\OO(T(n))$ for $T(n)=\OOm(1)$ contain all constants. \qed
\end{proof}
\begin{corollary}
The problem $\halt{\OO(1)}$ is undecidable.
\end{corollary}

To rephrase the corollary, there is no algorithm that would tell us which multi-tape Turing machines run in constant time and which do not. The catch is of course in big $\OO$ notation, which hides the constant. 

\section{One-tape machines}
	\label{one}

This section includes the results about decidability of problems $\halt{\FF}^1$ and $\halt{T(n)}^1$ for different $\FF$ and $T$. We will see that these results are more complex than their multi-tape analogs which arises from the  ability of multi-tape Turing machines to efficiently count steps while simulating another Turing machine.

\subsection{Negative results}

Let us start with a technical lemma that will ease reductions of the problem $\halto$ to the problems concerning time complexity. It's multi-tape analog is Lemma~\ref{logLema2}.
\begin{lemma}
	\label{logLema}
Let $T:\NN\rightarrow\RPZ$ be a function such that $T(n)=\OOm(n\log n)$ and, for all $n\in\NN$, it holds $T(n)\geq n+1$. Then there exists an algorithm that takes as input a one-tape Turing machine $H$ and returns a one-tape Turing machine $\tilde{H}$ such that
$$H(\epsilon)=\infty\textrm{ \textbf{iff} }\tilde{H}\textrm{ runs in time }T(n)\textrm{ \textbf{iff} }\tilde{H}\textrm{ always halts.}$$
\end{lemma}
\begin{proof}
Because $T(n)=\OOm(n\log n)$, there exist constants $C,n_0\in\NN$ such that $6\leq C\leq n_0$ and, for all $n\geq n_0$, it holds $$T(n)\geq 3 n\log_C n+ 6n +1.$$

For an arbitrary one-tape Turing machine $H$ with tape alphabet $\Gamma$ and blank symbol $\blank$, let us describe a new one-tape Turing machine $\tilde{H}$:
\begin{itemize}
\item $\tilde{H}$ has the same input alphabet and blank symbol as $H$, but its tape alphabet is $\Gamma\bigcup\Gamma'\bigcup\{0,1,\#\}$, where $\Gamma'=\{a'; a\in\Gamma\}$.
Without loss of generality we can assume that sets $\{0,1,\#\}$, $\Gamma$ and $\Gamma'$ are pairwise disjoint.
\item On input $w$ of length $n$, $\tilde{H}$ first reads the input and if $n< n_0$, accepts in $(n+1)$ steps. If $n\geq n_0$, then  $\tilde{H}$ overwrites the input with $$\#1^{n-1}\#,$$
 leaving the head above the last written one. This all can be done in \textbf{$\mathbf{(n+1)}$ steps}.
\item $\tilde{H}$ will never again write or overwrite the symbol $\#$, which will serve as the left and right border for the head. From now on, the head will move exactly from the right $\#$ to the left $\#$ and vice versa. Thus we can only count how many times the head will pass from one $\#$ to another and multiply the result with $n$ to get how many steps were done. A transition of the head form one $\#$ to another will be called a \emph{(head) pass}.
\item For $m=\lceil  \log_C n\rceil$, $\tilde{H}$ can transform its tape into $$\#\blank^{m}\blank'\blank 0^{n-3-m}\#$$  in the next $\mathbf{(2m+2)}$\textbf{ head passes}\footnote{Note that $C^{m-1}\leq n-1<C^m$.}.

This can be done if on each pass to the right $\tilde{H}$ turns $(C-1)$ successive ones into zeros, leaves out the next symbol one, turns the next $(C-1)$ successive ones into zeros \ldots until it comes to $\#$. Also when passing to the right, it adds another blank symbol after the last written blank symbol. When passing to the left it changes nothing. When there are no more ones, it makes two additional passes to insert $\blank'\blank$ after blank symbols. 

 Until this point we did not need any information about $H$.
\item The tape is now prepared for simulation of $H$ on input $\epsilon$. The symbols from $\Gamma$ tell us how the tape of $H$ looks like and the (only) symbol from $\Gamma'$ tells us the current head position in $H$.
\item The simulation goes as follows: in each pass, $\tilde{H}$ turns $(C-1)$ successive zeros into blank symbols, leaves out the next symbol zero, turns next $(C-1)$ successive zeros into blank symbols \ldots and when the head comes to the ``simulation part'' of the tape, it simulates one step of $H$ iff the head of $H$ would move in the same direction as the head of $\tilde{H}$ is currently moving. Thus $\tilde{H}$ simulates at least one and at most two steps of $H$ in two head passes. 
\item If $\tilde{H}$ runs out of zeros on its tape before the simulation of $H$ has finished, it halts (e.g. goes in $\qa$). Else,  $\tilde{H}$ starts an infinite loop so that $\tilde{H}(w)=\infty$.

\item From $6\leq C\leq n_0\leq n$ it follows that $n-3-m\geq C^{m-2}$, so $\tilde{H}$ needs at least $(m-1)$ head passes to erase all zeros.

Thus if  $\tilde{H}$ halts, this means that $H$ does not complete its computation on input $\epsilon$ in $\lfloor\frac{m-1}{2}\rfloor$ steps. In this case $\tilde{H}$ makes at most $\mathbf{m}$\textbf{ head passes} from the beginning of the simulation until it halts and thus makes at most $T(n)$ steps altogether on input $w$.

If $\tilde{H}$ does not halt, this means that $H$ halts on input $\epsilon$.
\end{itemize}

Note that since $m=\OOm(\log n)\neq\OO(1)$ it holds 
$$H(\epsilon)=\infty\textrm{ \textbf{iff} }\tilde{H}\textrm{ runs in time }T(n)\textrm{ \textbf{iff} }\tilde{H}\textrm{ always halts.}$$

To sum up, we have described a desired one-tape Turing machine $\tilde{H}$ and it is clear from the description that there exists an algorithm that constructs it from $H$. 
\qed
\end{proof}

The next theorem with its corollaries is an analog of Theorem~\ref{multiO} for multi-tape Machines.

\begin{theorem}
	\label{oneO}
Let $\FF\subseteq\{f:\NN\rightarrow\RPZ\}$ be a class of functions that contains arbitrary large constants. Then the problem $\halt{\FF}^1$ is undecidable.
\end{theorem}
\begin{proof}
The proof is by contradiction, so let a multi-tape Turing machine $H_{\FF}$ decide $\halt{\FF}^1$ for some class $\FF$ of functions that contains arbitrary large constants. Define the class  $F=\{T\in \FF,\ T(n)\geq n+1\textrm{ for all }n\}$ and consider two separate cases:
\begin{enumerate}
\item If for all functions $T\in F$, it holds $T(n)\neq\OOm(n\log n)$, then the following algorithm solves $\halto$:
\begin{itemize}
\item Given a one-tape Turing machine  $H$, construct a one-tape Turing machine $\tilde{H}$ that on input $w$
\begin{itemize}
\item simulates $|w|$ steps of ($H$ on input $\epsilon$).
\item If $H$ halts in less than $|w|$ steps, then $\tilde{H}$ also halts and does not make any additional steps. 
\item Else, $\tilde{H}$ makes at least additional $|w|\log |w|$ arbitrary steps and halts. 
\end{itemize}

This can easily be done, for example, by using the input portion of the  tape only for counting steps and the left portion of the tape for simulation of ($H$ on input $\epsilon$). We do not need $\tilde{H}$ to efficiently simulate $H$, but it is necessary that $\tilde{H}$ runs in constant time if $H$ halts on $\epsilon$.
\item Return $H_\FF(\tilde{H})$.
\end{itemize}

It is clear from the algorithm that $H$ halts on input $\epsilon$ \textbf{iff} $\tilde{H}$ runs in constant time \textbf{iff} $\tilde{H}$ runs in time $\tilde{T}(n)$ for some $\tilde{T}(n)\neq\OOm(n\log n)$. 

Now if $H_\FF(\tilde{H})=1$, then there exists a function $T\in \FF$, such that $\tilde{H}$ runs in time $T(n)$. If $T\in F$, then $T(n)\neq\OOm(n\log n)$, thus $H$ halts on input $\epsilon$. If $T\not\in F$, then by the definition of $F$ there exists $n_0$ such that $T(n_0)< n_0+1$, which implies that $\tilde{H}$ never reads the $(n_0+1)$-st bit of any input (see the proof of Lemma~\ref{pozMulti}). It follows that $\tilde{H}$ runs in constant time and consequentially $H$ halts on input $\epsilon$.

 If $H_\FF(\tilde{H})=0$, then $\tilde{H}$ does not run in constant time because $\FF$ contains arbitrary large constants. So $H$ does not halt on input $\epsilon$.

It follows that the above algorithm decides $\halto$, which is a contradiction.
\item If for some function $T\in F$ holds $T(n)=\OOm(n \log n)$, then we can (again) solve $\halto$:

For an arbitrary one-tape Turing machine $H$, use Lemma~\ref{logLema} to construct a one-tape Turing machine $\tilde{H}$ which runs in time $T(n)$ \textbf{iff} $\tilde{H}$ always halts \textbf{iff} $H(\epsilon)=\infty$.
Now $\left(1-H_{\FF}(\tilde{H})\right)$ tells whether  $H$ halts on input $\epsilon$.

Since $\halto$ is undecidable, we have a contradiction.

\end{enumerate}

To sum up, $\halt{\FF}^1$ is undecidable.\qed
\end{proof}

\begin{corollary}
	\label{cor1}
The problems  $\halt{\PP}^1$, $\halt{\EXP}^1$ and $\halt{\OO(T(n))}^1$ for $T(n)=\OOm(1)$ are undecidable.
\end{corollary}
\begin{proof}
The classes $\PP$, $\EXP$ and $\OO(T(n))$ for $T(n)=\OOm(1)$ contain all constants. \qed
\end{proof}
\begin{corollary}
	\label{cor2}
The problem $\halt{\OO(1)}^1$ is undecidable.
\end{corollary}
Now let us look at decidability of problems $\halt{T(n)}^1$, where $T$ is an explicit function. While Theorem~\ref{iff} was enough for the multi-tape case, we need Theorem~\ref{neg1}, Proposition~\ref{poz1} and Theorem~\ref{main} to answer questions about decidability of $\halt{T(n)}^1$.

\begin{theorem}
	\label{neg1}
Let $T:\NN\rightarrow\RPZ$ be a function such that $T(n)=\OOm(n\log n)$ and, for all $n\in\NN$, it holds $T(n)\geq n+1$. Then the problem $\halt{T(n)}^1$ is undecidable.
\end{theorem}
\begin{proof}
For an arbitrary one-tape Turing machine $H$, let $\tilde{H}$ be a one-tape Turing machine that runs in time $T(n)$ iff $H(\epsilon)=\infty$. By Lemma~\ref{logLema} we can construct  it from $H$.

Now if the problem $\halt{T(n)}^1$ was decidable, we could decide whether  $\tilde{H}$ runs in time $T(n)$ and thus also whether $H$ halts on input $\epsilon$. This is a contradiction since $\halto$ is undecidable.\qed
\end{proof}

\subsection{Positive results}

The next proposition shows that the condition $T(n)\geq n+1$ in Theorem~\ref{neg1} is inevitable. The proof is just the same as in the multi-tape case (Lemma~\ref{pozMulti}) an will be left out.

\begin{proposition}
	\label{poz1}
Let $T:\NN\rightarrow\RPZ$ be a function such that, for some $n_0\in\NN$, it holds $T(n_0)< n_0+1$. Then the problem $\halt{T(n)}^1$ is decidable.
\end{proposition}

%

Until this point we know that, for $T(n)=\OOm(n\log n)$, we can solve $\halt{T(n)}^1$ iff, for some $n_0\in\NN$, it holds $T(n_0)< n_0+1$. It remains to see that $\halt{T(n)}^1$ is decidable for all nice functions $T(n)=\oo(n\log n)$. 
\begin{theorem}
	\label{main}
For any manageable function $T:\NN\rightarrow\RR^+$, for which the function $\frac{n\log n}{T(n)}$ computably converges to $\infty$, the problem $\halt{T(n)}^1$ is decidable.
\end{theorem}

Some examples of manageable functions were given in Proposition~\ref{examples}. Note that Theorem~\ref{main} tells us that we can algorithmically verify not only whether a one-tape Turing machine runs in time $Cn+D$ for constants $C,D\in\NN$, but also whether it runs in time like $(n+1) \sqrt{\log (n+2)}$. The proof of this theorem will be given after we introduce the notion of crossing sequences, which will be the main tool in the proof.

\paragraph{Crossing sequences.}
For a one-tape Turing machine $M$, we can number the cells of its tape with integers so that the cell 0 is the one where $M$ starts its computation. Using this numbering we can number the boundaries between cells as shown on Figure~\ref{tape}. Whenever we say that an input is written on the tape, we mean that its $i$-th bit is in cell $(i-1)$ and all other cells contain the blank symbol $\blank$.

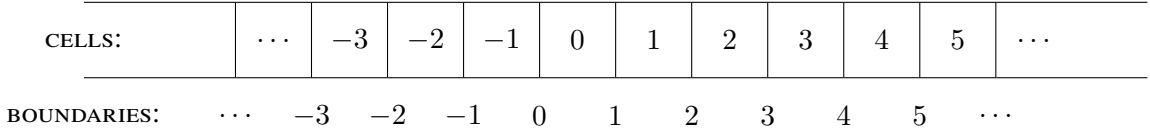
\begin{figure}[!htb]
	\begin{center}
	\begin{tikzpicture}
	\tikzstyle{prazno}=[inner sep=0pt, minimum size=0em]
		\node at ( -6,0) (c-50) [prazno] {};
		\node at ( -4,0) (c-40) [prazno] {};
		\node at ( -3,0) (c-30) [prazno] {};
		\node at (-2,0) (c-20)[prazno]  {};
		\node at ( -1,0) (c-10)[prazno]  {};
		\node at ( 0,0) (c00) [prazno] {};
		\node at ( 1,0) (c10) [prazno] {};
		\node at ( 2,0) (c20)  [prazno] {};
		\node at ( 3,0) (c30)[prazno]   {};
		\node at (4,0) (c40)[prazno]   {};
		\node at ( 5,0) (c50) [prazno]  {};
		\node at (6,0) (c60) [prazno]  {};
		\node at ( 8,0) [prazno]   {}
			edge [-] (c-50);
		\node at ( -6,-0.5) [prazno] {\sc{boundaries}:};
		\node at ( -4,-0.5)  [prazno] {$\ldots$ };
		\node at ( -3,-0.5)  [prazno] {$-3$};
		\node at (-2,-0.5) [prazno]   {$-2$};
		\node at ( -1,-0.5)[prazno]   {$-1$};
		\node at ( 0,-0.5)  [prazno] {$0$};
		\node at ( 1,-0.5) [prazno]  {$1$};
		\node at ( 2,-0.5)   [prazno] {$2$};
		\node at ( 3,-0.5) [prazno]   {$3$};
		\node at (4,-0.5)[prazno]   {$4$};
		\node at ( 5,-0.5)  [prazno]  {$5$};
		\node at (6,-0.5)  [prazno]  {$\ldots$};
		\node at ( 8,-0.5) [prazno]   {};
		\node at ( -6,1) (c-51) [prazno]  {};
		\node at ( -4,1) (c-41) [prazno]  {}
			edge [-] (c-40);
		\node at ( -3,1) (c-31) [prazno]  {}
			edge [-] (c-30);
		\node at (-2,1) (c-21)[prazno]   {}
			edge [-] (c-20);
		\node at ( -1,1) (c-11)  [prazno] {}
			edge [-] (c-10);
		\node at ( 0,1) (c01) [prazno] {}
			edge [-] (c00);
		\node at ( 1,1) (c11) [prazno]  {}
			edge [-] (c10);
		\node at ( 2,1) (c21) [prazno]  {}
			edge [-] (c20);
		\node at ( 3,1) (c31) [prazno]  {}
			edge [-] (c30);
		\node at (4,1) (c41)[prazno]  {}
			edge [-] (c40);
		\node at ( 5,1) (c51) [prazno]  {}
			edge [-] (c50);
		\node at (6,1) (c61) [prazno]  {}
			edge [-] (c60);
		\node at ( 8,1)  [prazno]  {}
			edge [-] (c-51);
		\node at ( -6,0.5) [prazno] {\sc{cells}:};
		\node at ( -3.5,0.5)  [prazno] {$\ldots$ };
		\node at ( -2.5,0.5)  [prazno] {$-3$};
		\node at (-1.5,0.5) [prazno]   {$-2$};
		\node at ( -0.5,0.5)[prazno]   {$-1$};
		\node at ( 0.5,0.5)  [prazno] {$0$};
		\node at ( 1.5,0.5) [prazno]  {$1$};
		\node at ( 2.5,0.5)   [prazno] {$2$};
		\node at ( 3.5,0.5) [prazno]   {$3$};
		\node at (4.5,0.5)[prazno]   {$4$};
		\node at ( 5.5,0.5)  [prazno]  {$5$};
		\node at (6.5,0.5)  [prazno]  {$\ldots$};
		\node at ( 8,0.5) [prazno]   {};
	\end{tikzpicture}
	\end{center}
	\caption{Numbering of tape cells and boundaries of a one-tape Turing machine.}
	\label{tape}
\end{figure}

Let $\tau$ be the tape of $M$ with some symbols written on it. We can cut the tape on finitely many boundaries to get \emph{tape segments} $\tup{\tau_1}{\tau_2}{\tau_k}$ so that $\tau=\tau_1\tau_2\cdots\tau_k$, where $\tau_1$ is left infinite, $\tau_k$ is right infinite and the other segments are finite. We can also start with a tuple of segments $\tup{\tilde{\tau}_1}{\tilde{\tau}_2}{\tilde{\tau}_l}$ and glue them together to get a tape 
$\tilde{\tau}=\tilde{\tau}_1\tilde{\tau}_2\cdots\tilde{\tau}_l$. This can be done if $\tilde{\tau}_1$ is left infinite, $\tilde{\tau}_l$ is right infinite, other segments are finite and exactly one of the segments has a prescribed location for the cell 0 (where $M$ starts its computation). If we run $M$ on $\tilde{\tau}$, we have three possible results of the computation:
\begin{itemize}
\item $M$ stops in state $\qa$. We say that $M$ accepts the tape $\tilde{\tau}$ and write $M(\tilde{\tau})=1$,
\item $M$ stops in state $\qr$. We say that $M$ rejects the tape $\tilde{\tau}$ and write $M(\tilde{\tau})=0$,
\item $M$ does not halt. We say that $M$ runs forever and write $M(\tilde{\tau})=\infty$.
\end{itemize}


Suppose that $M$ crosses the $i$-th boundary of tape $\tau$ at steps $\seq{t_1}{t_2}$ (this sequence can be finite or infinite). If $M$ was in state $q_j$ after making the step $t_j$, for all $j$, then we say that $M$ produces the \emph{crossing sequence} $C=\seq{q_1}{q_2}$ on the $i$-th boundary of $\tau$ and we denote its length by $|C|\in\NN\bigcup\{\infty\}$. Note that this sequence contains all the information the machine carries across the $i$-th boundary of the tape, thus the next proposition proven by Hennie~\cite{hennie} is very intuitive.
\begin{proposition}
	\label{hennie1}
Let $\tau_1\tau_2$ and $\tilde{\tau}_1\tilde{\tau}_2$  be two tapes of a one-tape Turing machine $M$ such that $M(\tau_1\tau_2)=M(\tilde{\tau}_1\tilde{\tau}_2)$. Suppose that segments $\tau_1$ and $\tau_2$ are joined at boundary $i>0$ and segments $\tilde{\tau}_1$ and $\tilde{\tau}_2$ are joined at boundary $j>0$. If $M$ on tape $\tau_1\tau_2$ on boundary $i$ produces the same crossing sequence as on tape $\tilde{\tau}_1\tilde{\tau}_2$ on boundary $j$, then
\begin{enumerate}[ a)]
\item $M(\tau_1\tau_2)=M(\tilde{\tau}_1\tilde{\tau}_2)=M(\tau_1\tilde{\tau}_2)=M(\tilde{\tau}_1\tau_2)$,
\item the crossing sequences generated on tapes $\tau_1\tau_2$ and $\tau_1\tilde{\tau}_2$ at corresponding boundaries of segment $\tau_1$ are identical,
\item the crossing sequences generated on tapes $\tau_1\tau_2$ and $\tilde{\tau}_1\tau_2$ at corresponding boundaries of $\tau_2$ are identical.
\end{enumerate}
\end{proposition}
Note that conditions $i>0$ and $j>0$  cause the cell 0 to be in left segments, which makes it possible to swap right segments. The same result also holds if we put $i,j\leq0$. The following corollary is now trivial.
\begin{corollary}
	\label{hennie2}
Let $\tau_1\tau_2\tau_3$ be  a tape of a one-tape Turing machine $M$. Suppose that segments $\tau_1$ and $\tau_2$ are joined at boundary $i>0$ and segments $\tau_2$ and $\tau_3$ are joined at boundary $j$. If $M$ on tape $\tau_1\tau_2\tau_3$ produces the same crossing sequences on boundaries $i$ and $j$, then $M(\tau_1\tau_2\tau_3)=M(\tau_1(\tau_2)^n\tau_3)$ for all $n\in\NN$.
\end{corollary}

Note that condition $i>0$ could be replaced by $j\leq 0$. In other words, if the same crossing sequence appears on both ends of some tape segment that does not contain cell 0, then we can remove this segment or add extra copies of it next to each other without affecting the result of the computation.


The next lemma, implicit in Kobayashi~\cite{kobayashi}, is not that intuitive.
\begin{lemma}
	\label{kobi}
Let $T:\NN\rightarrow\RR^+$ be a function such that $T(n)=\oo(n \log n)$ and let
$$g(n)=\left\{
	\begin{array}{lcl}
		\frac{n \log n}{T(n)}&;& n\geq 2 \\
		1 &;& n=0,1.
	\end{array}
\right.$$
Then, for any integer $q\geq 2$, there exists a constant $c$ such that any one-tape Turing machine with $q$ states that runs in time $T(n)$, on each input produces only crossing sequences of lengths bounded by $c$. What is more, $c$ can be any constant satisfying $c\geq \max\{T(0),T(1)\}$ and the following inequality:
\begin{equation}
	\label{eq}
3\frac{q n^{(\log q)/g(n)^{1/2}}-1}{q-1}\leq n-3-\frac{n}{g(n)^{1/2}}+c\frac{g(n)^{1/2}}{\log n}
\end{equation}
for all $n\geq 2$.
\end{lemma}

Note that since $\displaystyle\lim_{n\to \infty} g(n)=\infty$, then, for any $q\geq2$, there exists a constant $c\geq \max\{T(0),T(1)\}$ such that Inequality~\eqref{eq} holds for all $n\geq 2$. Now let us prove the lemma.
\begin{proof}
Let $M$ be a one-tape Turing machine with $q$ states that runs in time $T(n)$. Let $c\geq \max\{T(0),T(1)\}$ be such that Inequality~\eqref{eq} holds for all $n\geq 2$ and suppose that $M$ produces a crossing sequence of length more than $c$ on some input. Let $w$ be the shortest such input and let $n_0=|w|$. Note that $n_0\geq 2$ since $c\geq \max\{T(0),T(1)\}$. Suppose $w$ was given to $M$.

Let $h$ be the number of boundaries from $\{\tup{1}{2}{n_0-1}\}$ on which crossing sequences of lengths less than $(\log n_0)/g(n_0)^{1/2}$ were produced. Then we have 
$$\frac{n_0\log n_0}{g(n_0)}=T(n_0)>c+(n_0-2-h)\frac{\log n_0}{g(n_0)^{1/2}}$$
and hence
\begin{align*}
h> &\ n_0-2-\frac{n_0}{g(n_0)^{1/2}}+c\frac{g(n_0)^{1/2}}{\log n_0}\\
\geq  &\ 3\frac{q n_0^{(\log q)/g(n_0)^{1/2}}-1}{q-1}+1\\
= &\ 3\frac{q^{(\log n_0)/g(n_0)^{1/2}+1}-1}{q-1}+1.
\end{align*}

Moreover, a simple counting shows that there are at most $(q^{(\log n_0)/g(n_0)^{1/2}+1}-1)/(q-1)$ distinct crossing sequences of lengths less than $(\log n_0)/g(n_0)^{1/2}$.

Hence, by the pigeonhole principle, there exist at least four boundaries  in $\{\tup{1}{2}{n_0-1}\}$ on which the same crossing sequence $s$ was produced. Now if a crossing sequence of length more than $c$ was produced on some boundary $i\in\ZZ$, we can find two boundaries in $\{\tup{1}{2}{n_0-1}\}$ on which $s$ was produced, such that $i$ does not lie between them. If we cut away the subword of $w$ between those two boundaries, we get an input for $M$ of length less than $n_0$ on which $M$ produces a crossing sequence of length more than $c$. This contradicts the selection of $w$ and completes the proof of the lemma.\qed
\end{proof}
 
This lemma has some interesting consequences.



\begin{corollary}
	\label{space}
If a one-tape Turing machine $M$ runs in time $T(n)=\oo(n \log n)$, then there exists a constant $D$, such that $M$ on each input $w$ visits at most $|w|+D$ cells.
\end{corollary}

\begin{proof}
After Lemma~\ref{kobi}, the length of crossing sequences produced by $M$ is bounded by a constant and thus $M$ produces only constantly many crossing sequences. If $K$ is this constant, let us prove that $M$ visits at most $K$ cells to the right of the input.

Suppose that this is not true for some input $w$. If we run $M$ on $w$, then there are at least two boundaries with index greater than or equal to $|w|$, say $i$ and $j$, that produce the same non-empty crossing sequence. At the beginning of the computation we only have blank symbols between those two boundaries, thus all boundaries $i+k|j-i|$ for $k\in\NN$ produce the same crossing sequence. This is a contradiction with $M$ running in finite time, thus $M$ visits at most $K$ cells to the right of the input. 

The same way we can show that $M$  visits at most $K$ cells to the left of the input, which completes the proof.\qed
\end{proof}

The next corollary gives us even more feeling about the $\oo(n \log n)$ bound on running time for a one-tape Turing machine.

\begin{corollary}
	\label{time}
If a one-tape Turing machine runs in time $\oo(n \log n)$, then it runs in linear time.
\end{corollary}

\begin{proof}
Let $M$ be a one-tape Turing machine that runs in time $\oo(n \log n)$. The main observation is that $M$ on each input $w$ halts exactly after $\sum |C|$ steps, where the sum is over all crossing sequences $C$ produced on boundaries of the tape.
From Lemma~\ref{kobi} it follows that each addend is bounded by a constant and from Corollary~\ref{space} it follows that there are at most $|w|+D$ of them for some constant $D$.\qed
\end{proof}

The following lemma makes an introduction to the proof of Theorem~\ref{main}.
\begin{lemma}
	\label{construct}
Let $T:\NN\rightarrow\RR^+$ be a function for which  $\lfloor T\rfloor$ is computable and the function 
$$g(n)=\left\{
	\begin{array}{lcl}
		\frac{n \log n}{T(n)}&;& n\geq 2 \\
		1 &;& n=0,1
	\end{array}
\right.$$
computably converges to $\infty$. Then given $q\in\NN$, we can compute a constant upper bound of the length of the crossing sequences produced by any $q$-state one-tape Turing machine that runs in time $T(n)$.
\end{lemma}
\begin{proof}
After Lemma~\ref{kobi}, we only need to construct a constant $c\geq  \max\{T(0),\ T(1)\}$ which satisfies Inequality~\eqref{eq}
for the given $q$ and all $n\geq 2$. The construction of $c$ can go as follows:
\begin{itemize}
\item Use computable convergence of $g$ to find $N \in\NN$ such that for all $n\geq N$ holds $g(n)\geq 4(\log q)^2$. Increase $N$ if necessary so that, for all $n\geq N$, it also holds $\sqrt{n}\leq\frac{1}{12}n$ and $g(n)\geq 16$.

It is easy to see that Inequality~\eqref{eq} holds for all $n\geq N$ independently of the value of $c\geq 0$.
\item Use computability of $\lfloor T\rfloor$ to find such $c\in\NN$ that the Inequality~\eqref{eq} holds for $2\leq n< N$.

Note that $g(n)\geq \frac{1}{\lfloor T(n)\rfloor+1}$ for $n\geq 2$.
\item Increase $c$ to get $c\geq \max\{ T(0),\ T(1)\}$. \qed
\end{itemize}
\end{proof}

\paragraph{Proof of Theorem~\ref{main}.}
Let  $T:\NN\rightarrow\RR^+$ be a manageable function for which the function $\frac{n\log n}{T(n)}$ computably converges to $\infty$. Because $T$ is manageable, $\lfloor T\rfloor$ is computable after Proposition~\ref{compf}. The following algorithm verifies whether a one-tape Turing machine $M$ with $q$ states runs in time $T(n)$, thus solving $\halt{T(n)}^1$.
\begin{enumerate}[(i)]
\item  Use Lemma~\ref{construct} to \step{construct an upper bound $c\in\NN$ on the length of the crossing sequences produced by  any one-tape Turing machine with $q$ states that runs in time $T(n)$.}
\item  \step{Compute} $$K=\frac{q^{c+1}-1}{q-1},$$ which is an upper bound on the number of distinct crossing sequences produced by any one-tape Turing machine with $q$ states that runs in time $T(n)$.

Now the idea is as follows. If $M$ runs in time $T(n)$, then we can (cleverly) split each input into parts of lengths at most $K$, on which $M$ produces pairwise distinct crossing sequences. Then we only have to worry about how many steps $M$ makes on each of those parts.
\item \step{Define sets $X=\emptyset$ and $S=\emptyset$.}

The following steps are defined in such a way that, if $M$ runs in time $T(n)$, then 
\begin{itemize}
\item at the end of Step~(\ref{stepX}) the set $X$ will contain all inputs $x$ for which $M$ produces pairwise distinct crossing sequences on boundaries $\{i;\ 0<i\leq|x|\}$ and
\item at the end of Step~(\ref{stepY}) the set $S$ will contain all crossing sequences produced by $M$ on all inputs $w$ on boundaries $\{i;\ 0<i\leq|w|\}$.
\item Note that the cardinality of $S$ and the length of each $x\in X$ will be bounded by $K$.
\end{itemize}
\item 
	\label{stepX}
For each input $w$ of length $|w|\leq K$, simulate $M$ on $w$ and
\begin{itemize}
	\item if $M$ makes more than $T(|w|)$ steps, \textbf{return 0}. Else,
	\item add all crossing sequences produced on boundaries $\{i;\ 0<i\leq|w|\}$ to $S$ and
	\item if these crossing sequences are pairwise distinct, add $w$ to $X$ and
	\item define $T_w$ as the number of steps $M$ makes on input $w$.
\end{itemize}
\item For each crossing sequence $s\in S$, define the set $Y_s=\emptyset$.

We say that a word $y$ is \emph{compatible} with $s$, if there exists an input $w$ such that $w=w_1yw_2$ and $M$ produces   $s$ on boundaries $|w_1|>0$ and $|w_1|+|y|$. Additionally, if the crossing sequences produced on boundaries $\{i;\ |w_1|<i\leq|w_1|+|y|\}$ are pairwise distinct, we say that $y$ is \emph{primitive compatible} with $s$. Note that the definition does not depend on the choice of words $w_1$ and $w_2$.

If $M$ runs in time $T(n)$, then the set $Y_s$ will at the end of Step~(\ref{stepY}) contain all non-empty words that are primitive compatible with $s$.
\item  
	\label{stepY}
For each crossing sequence $s\in S$ and each word $y$ of length $0<|y|\leq K$,  if $y$ is primitive compatible with $s$, put $y$ in $Y_s$. For such a $y$, if $w_1yw_2$ is some input that induces the compatibility of $y$ with $s$, put all crossing sequences that $M$ produces on input $w_1yw_2$ on boundaries $\{i;\ |w_1|<i<|w_1|+|y|\}$ in $S$.

We can realize this step with a simple simulation of $M$ on word $y$, by which we pretend that on the left and right of $y$ are some words that collaborate in the production of crossing sequence $s$ on both ends of $y$:
%
\begin{itemize}
	\item If $s$ is the empty sequence, then $y$ is primitive compatible with $s$ iff $y$ is a one-symbol word. Else, let $s=\tup{q_1}{q_2}{q_k}$.
	\item Suppose the input $y$ is written on $M$'s tape while $M$ is in state $q_1$ and has its head above cell 0.
	\item Simulate $M$ until it leaves the portion of the tape where $y$ was written.
\begin{itemize}
		\item If a crossing sequence of length greater than $c$ is produced, $M$ does not run in time $T(n)$, thus \textbf{return 0}.
		\item If $M$ halted before its head crossed boundaries $0$ or $|y|$, then $y$ is not compatible with $s$.
		\item If $M$'s head crossed boundary $0$ in a state $\neq q_2$ or boundary $|y|$ in a state $\neq q_1$, then $y$ is not compatible with $s$.
\end{itemize}
	\item If $M$'s head crossed boundary $0$ in state $q_2$, put its head back on cell $0$, change its state to $q_3$ and continue with simulation. If there is no state $q_3$ (\ie $k=2$), then $y$ is not compatible with $s$.
	\item If $M$'s head crossed boundary $|y|$ in state $q_1$, then 
\begin{itemize}
		\item if $k>1$, change $M$'s state to $q_2$, put its head on cell $(|y|-1)$ and continue with simulation.
		\item if $k=1$, stop the simulation and if crossing sequences produced on boundaries $0<i<|y|$ are pairwise distinct and none is equal to $s$, then $y$ is primitive compatible with $s$, otherwise $y$ is not primitive compatible with $s$.
\end{itemize}
	\item Continue with the simulation as just described until it is found whether $y$ is primitive compatible with $s$ or not. If it is, add the crossing sequences produced on boundaries $\{i;0<i<|y|\}$  during simulation to $S$. 
\end{itemize}

For $y\in Y_s$, define $T_{s,y}=|s|+\sum |C|$ where the sum is over all crossing sequences $C$ produced on boundaries $\{\tup{1}{2}{|y|-1}\}$ in the above simulation. It follows that $T_{s,y}$ is the time $M$ spends on part $y$ of any input that induces compatibility of $y$ with $s$.

Note that $S$ can acquire new elements during this simulation, thus we have to construct $Y_s$ for those elements too. Since $S$ contains only crossing sequences of length at most $c$, the procedure is finite.
\item 
	\label{impl}
Check whether all non-empty words can be constructed in the following way:
\begin{itemize}
	\item Start with some word $w\in X$. 
	\item Repeat the next line finitely many times (possibly 0):
	\begin{itemize}
		\item Suppose $M$ on input $w$ produces a crossing sequence $s$ on some boundary $i$, $0<i\leq|w|$. Insert some word from $Y_s$ into $w$ on the place where $s$ was produced.
	\end{itemize}
\end{itemize}

A possible implementation of this step is given later in Lemma~\ref{implementation}. If there exists some non-empty word that cannot be constructed this way, \textbf{return 0}.

It is important to see that, if $M$ runs in time $T(n)$, then each non-empty word $w$ can be constructed this way. This is because such $M$ can generate at most $K$ different crossing sequences on input $w$ and, if $w\not\in X$, then at least two crossing sequences produced on boundaries $\{\tup{1}{2}{|w|}\}$ must be the same. Hence, we can find and cut out a section of $w$ that is primitive compatible with some crossing sequence produced on both of section's ends.
If we continue with cutting such sections out, we are left with a word $x\in X$ at the end and thus each section that was cut out belongs to some $Y_s$.
So we can find all ``parts'' of $w$ in sets $Y_s$ and $X$ and we can construct $w$ as described above.
\item If the number of steps $M$ makes on inputs $X$ and parts $Y_s$ is appropriate for $M$ to run in time $T(n)$, then \textbf{return 1}, else \textbf{return 0}.

This step can be implemented in the following way:
\begin{itemize}
	\item For each $x\in X$ and for all choices of subsets $\tilde{Y}_s\subseteq Y_s$, for which we can construct some word as in Step~(\ref{impl}) by starting with $x$ and inserting only words from sets $\tilde{Y}_s$ on appropriate places such that each word from each $\tilde{Y}_s$ is used exactly once,
%
\begin{itemize}
		\item use manageability of $T$ to check whether the inequality
$$T_x+\sum_{s\in S\vphantom{\tilde{Y}_s}}\sum_{y\in \tilde{Y}_s}T_{s,y}\cdot k_{s,y}\leq T\left(|x|+\sum_{s\in S\vphantom{\tilde{Y}_s}}\sum_{y\in \tilde{Y}_s}|y|\cdot k_{s,y}\right)$$
holds for all $k_{s,y}\in\NN\backslash\{0\}$. If it does not, \textbf{return 0}.
\end{itemize}
\item \textbf{Return 1}.
\end{itemize}

Note that in the argument of $T$ on the right-hand side of the inequality, we have the length of some word constructed as in Step~(\ref{impl}) by starting with $x$ and inserting $k_{s,y}$ words $y\in Y_s$ on appropriate places. On the left-hand side we have the number of steps that $M$ makes on such an input. Step~(\ref{impl})  tells that all non-empty inputs are considered this way and the condition before the inequality tells that it is possible to use only and all parts from sets  $\tilde{Y}_s$ at once.
\end{enumerate}

The comments inside the algorithm show its finiteness and correctness. The only thing missing is the implementation of Step~(\ref{impl}), which is described next.

Suppose $X$, $S$ and $Y_s$ are as at the end of Step~(\ref{stepY}). For each $\tilde{S}\subseteq S$ and $s\in \tilde{S}$, define the language
\begin{align*}
L_{s,\tilde{S}}=\{&\textrm{words compatible with }s\textrm{ on which $M$ produces only crossing sequences from }\tilde{S}\\
&\textrm{on the simulation described in Step~(\ref{stepY})}\}.
\end{align*}

In the next few results we use some well known facts about regular expressions, regular languages and finite automata. All the material needed can be found \eg in~\cite[Chapter 1]{sipser}. Also, the Proposition~\ref{regbasic} will prove useful.
\begin{lemma}
	\label{regular}
There exists an algorithm which can be used as a subroutine in  Step~(\ref{impl}) that, given $\tilde{S}\subseteq S$ and $s\in \tilde{S}$, it constructs a deterministic finite automaton that recognizes  $L_{s,\tilde{S}}$.
\end{lemma}
\begin{proof}
We will prove this by induction on the size of $\tilde{S}$. If $\tilde{S}=\{s\}$, then 
$$L_{s,\tilde{S}}=\left\{
	\begin{array}{lcl}
		\{\tup{y_1}{y_2}{y_k}\}^*  &;&Y_s \textrm{ contains one-symbol words }\tup{y_1}{y_2}{y_k}; \\
		\emptyset  &;&Y_s \textrm{ contains no one-symbol word},
	\end{array}
\right.$$
and the lemma holds. 
If $\tilde{S}$ has more than one element, then
$$L_{s,\tilde{S}}=\left(\bigcup_{\substack{y_1y_2\cdots y_k\in Y_s\\ \tup{s_1}{s_2}{s_{k-1}}\in\tilde{S}}}\left\{y_1L_{s_1,\bar{S}}y_2L_{s_2,\bar{S}}\cdots L_{s_{k-1},\bar{S}}y_k\right\}\right)^*$$
where
\begin{itemize}
\item $\tup{y_1}{y_2}{y_k}$ are symbols of some word from $Y_s$, 
\item $\tup{s_1}{s_2}{s_{k-1}}$ are crossing sequences generated between $\tup{y_1}{y_2}{y_k}$ on the simulation described in Step~(\ref{stepY}), and
\item $\bar{S}=\tilde{S}\backslash\{s\}$.
\end{itemize}
By the induction hypothesis we can construct deterministic finite automata that recognize languages $L_{s_i,\bar{S}}$. After Proposition~\ref{regbasic} we can construct a deterministic finite automaton that recognizes $L_{s,\tilde{S}}$.\qed
\end{proof}

The next lemma finishes the proof of Theorem~\ref{main}.
\begin{lemma}
	\label{implementation}
We can implement Step~(\ref{impl}).
\end{lemma}
\begin{proof}
It is easy to see that we only need to give an algorithm which verifies whether all words are in the set
$$L=\bigcup_{x_1x_2\cdots x_k\in X}\left\{x_1L_{s_1,S}x_2L_{s_2,S}\cdots x_k L_{s_{k},S}\right\}\bigcup\{\epsilon\},$$
where $\tup{x_1}{x_2}{x_k}$ are symbols of some word from $X$ and $\tup{s_1}{s_2}{s_{k}}$ are crossing sequences generated on boundaries $\{\tup{1}{2}{k}\}$ on simulation of $M$ on input $x_1x_2\cdots x_k$.

We can use Lemma~\ref{regular} to construct a deterministic finite automaton $F$ that recognizes the language $L$. It follows that $L$ contains all words iff all reachable states of $F$ are accepting, which can algorithmically be verified.\qed
\end{proof}

The next theorem is an improvement over the result of Hartmanis~\cite{Hartmanis} which states that one-tape Turing machines that run in $\oo(n\log n)$ time accept only regular languages. Note that, because of Corollary~\ref{time}, one-tape Turing machines that run in $\oo(n\log n)$ time, run in linear time.
\begin{theorem}
	\label{constructive}
There exists an algorithm that takes integers $C,D\in\NN$ and a one-tape Turing machine $M$ as inputs and if $M$ runs in time $Cn+D$, returns an equivalent deterministic finite automaton.
\end{theorem}
\begin{proof}
Let integers $C,D\in\NN$ and a one-tape Turing machine $M$ with $q$ states be given. The case when $C=0$ is easy and is left for the reader. If $D=0$, then $M$ does not run in time $Cn+D$ because this would imply that $M$ makes no step on empty input. So we can suppose $C,D>0$.

Because of the simplicity of the function $T(n)=Cn+D$,
we can still use the manageability of $T$ and the computable convergence to $\infty$ of $\frac{n\log n}{T(n)}$, although $C$ and $D$ are given as parameters. Hence, we can verify whether $M$ runs in time $T(n)$ with the above algorithm. If the answer is yes, let $X$ and $S$ be the sets constructed after Step~(\ref{stepY}). Let $X_A\bigcup X_R$ be a partition of $X$ on inputs that are accepted by $M$ and on inputs that are rejected by $M$. Then $M$ accepts the language
$$L=\bigcup_{x_1x_2\cdots x_k\in X_A}\{x_1L_{s_1,S}x_2L_{s_2,S}\cdots x_k L_{s_{k},S}\}\bigcup E,$$
where 
\begin{itemize}\setlength{\itemsep}{-0mm}
\item $\tup{x_1}{x_2}{x_k}$ are symbols of some word from $X_A$,
\item $\tup{s_1}{s_2}{s_{k}}$ are crossing sequences generated on boundaries $\{\tup{1}{2}{k}\}$ on simulation of $M$ on input $x_1x_2\cdots x_k$, and
\item if $M$ accepts $\epsilon$, then $E=\{\epsilon\}$, else $E=\emptyset$.
\end{itemize}
We can use Lemma~\ref{regular} to construct a deterministic finite automaton that recognizes the language $L$.\qed
\end{proof}
\section{Acknowledgements}
The author wishes to thank his research advisor Sergio Cabello for many enlightening discussions and Andrej Bauer, Jurij Mihelič and Marko Petkov\v{s}ek for valuable comments.
%
%
%
\newpage
\bibliographystyle{abbrv}
\bibliography{literature}

\end{document}